\documentclass[iicol]{sn-jnl}

\jyear{2023}%

\theoremstyle{thmstyleone}%
\newtheorem{theorem}{Theorem}
\newtheorem{proposition}[theorem]{Proposition}%

\theoremstyle{thmstyletwo}%

\theoremstyle{thmstylethree}%
\newtheorem{definition}{Definition}%

\usepackage{booktabs,tabularx}
\usepackage{xfrac}
\usepackage[final]{microtype}
\usepackage{adjustbox}
\usepackage{amsfonts}
\usepackage{amsmath}
\usepackage{xcolor,colortbl}
\usepackage{bm}
\usepackage{graphicx}
\usepackage[english]{babel}
\usepackage{gensymb}
\usepackage{array}
\usepackage{dsfont}
\usepackage[round]{natbib}
\newcolumntype{C}[1]{>{\centering\arraybackslash}p{#1}}
\DeclareMathAlphabet{\mathcal}{OMS}{cmsy}{m}{n}

\def\eg{\emph{e.g.,\ }}
\def\etal{\emph{et al.}}
\def\ie{\emph{i.e.\ }}

\usepackage{hyperref}
\hypersetup{
    colorlinks=true,
    linkcolor=blue,
    urlcolor=cyan,
    }
\renewcommand{\qed}{\null\nobreak\hfill\rule{1ex}{1em}}
\newcommand{\QED}{}
\newcommand{\tetrahedron}{
  \mathchoice
    {\includegraphics[height=2.0ex]{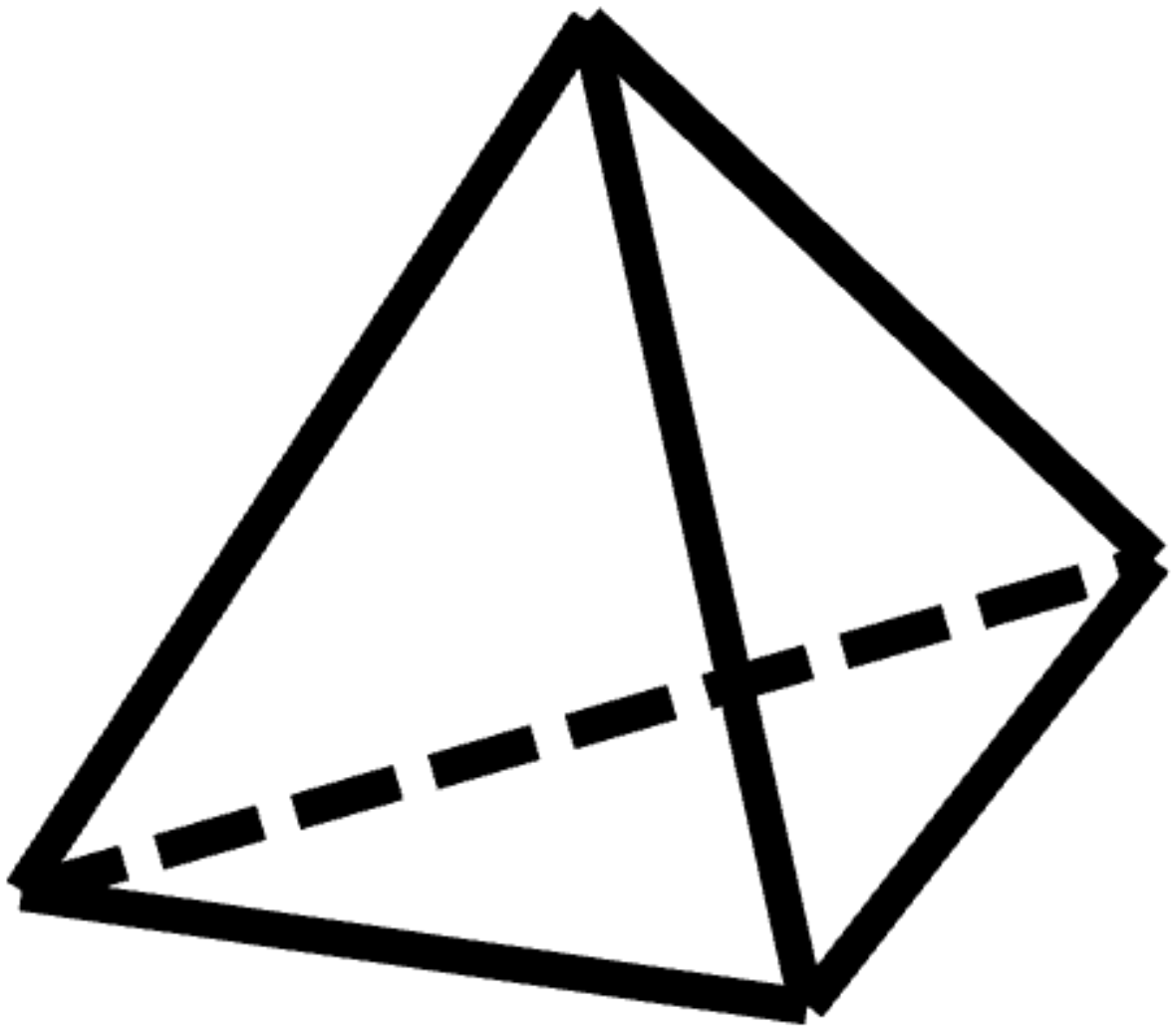}\,} 
    {\includegraphics[height=2.0ex]{Fig0.pdf}\,} 
    {\includegraphics[height=1.5ex]{Fig0.pdf}\,} 
    {\includegraphics[height=1ex]{Fig0.pdf}\,} 
}

\newcommand{\directline}[1]{{#1}}
\newcommand{\jacdet}{\vert J\vert}

\definecolor{Gray}{gray}{0.9}

\DisableLigatures[f]{encoding = *, family = *}

\begin{document}

\title[\hspace*{330pt}On Finite Difference Jacobian Computation in~Deformable~Image ~Registration]{On Finite Difference Jacobian Computation in~Deformable~Image ~Registration}


\author*[1]{\fnm{Yihao} \sur{Liu}}\email{yliu236@jhu.edu}

\author[2]{\fnm{Junyu} \sur{Chen}}\email{jchen245@jhmi.edu}

\author[1]{\fnm{Shuwen} \sur{Wei}}\email{swei14@jhu.edu}

\author[1]{\fnm{Aaron} \sur{Carass}}\email{\\aaron\_carass@jhu.edu}

\author[1]{\fnm{Jerry} \sur{Prince}}\email{prince@jhu.edu}

\affil[1]{\orgdiv{Department of Electrical and Computer Engineering}, \orgname{Johns Hopkins University}, \orgaddress{\city{Baltimore}, \state{Maryland}, \country{USA}}}

\affil[2]{\orgdiv{Russell H. Morgan Department of Radiology and Radiological Science}, \orgname{Johns Hopkins Medical Institutes}, \orgaddress{\city{Baltimore}, \state{Maryland}, \country{USA}}}


\abstract{
Producing spatial transformations that are diffeomorphic is a key goal in deformable image registration.
As a diffeomorphic transformation should have positive
Jacobian determinant~$\jacdet$ everywhere, the number of voxels with $\jacdet<0$ has been used to test for diffeomorphism and also to measure the irregularity of the transformation.
For digital transformations, $\jacdet$ is commonly approximated using a central difference, but this strategy can yield positive $\jacdet$'s for transformations that are clearly not diffeomorphic---even at the voxel resolution level.
To show this, we first investigate the geometric meaning of different finite difference approximations of $\jacdet$.
We show that to determine if a deformation is diffeomorphic for digital images, the use of any individual finite difference approximation of $\jacdet$ is insufficient.
We further demonstrate that for a 2D transformation, four unique finite difference approximations of $\jacdet$'s must be positive to ensure that the entire domain
is invertible and free of folding at the pixel level.
For a 3D transformation, ten unique finite differences
approximations of $\jacdet$'s are required to be positive.
Our proposed \textit{digital diffeomorphism} criteria solves several errors inherent in the central difference approximation of $\jacdet$ and accurately detects non-diffeomorphic digital transformations. The source code of this work is available at \url{https://github.com/yihao6/digital_diffeomorphism}.
}




\maketitle

\section{Introduction}
\label{s:intro}

The goal of deformable image registration is to establish a nonlinear spatial transformation that aligns two images. 
For many tasks, it is reasonable to assume that the anatomy in the two images share the same topology.
Therefore, for deformable registration algorithms, the ability to produce a topology preserving transformation is preferred.
Following the work of Christensen \etal~\citep{christensen1994deformable}, many registration algorithms either penalize~\citep{avants2008symmetric, chen2017mia, mok2020fast} or constrain~\citep{beg2005ijcv, chen2015boe, dalca2018unsupervised} their output transformations to be diffeomorphic and preserve topology.
In the continuous domain, a diffeomorphic transformation is a smooth and invertible mapping with a smooth inverse that is guaranteed to maintain the topology of the anatomy being transformed.
A diffeomorphic transformation should have positive Jacobian
determinant $\jacdet$ everywhere\footnote{We do not consider the case of all negative $\jacdet$, which would produce a reflection of the entire image.}, and testing if a transformation is diffeomorphic involves local computation of $\jacdet$.
When a transformation is not diffeomorphic, it is common to use the number of voxels with negative $\jacdet$~\citep{balakrishnan2019voxelmorph, chen2022transmorph, chen2021vit, liu2022coordinate} or the standard deviation of the logarithmic transformed $\jacdet$~\citep{hering2021learn2reg, kabus2009evaluation, johnson2002consistent} to measure the irregularity of the transformation.

Given a digital transformation that is defined on a regular grid, a widely accepted practice for computing the Jacobian is to use finite difference approximations of spatial derivatives.
There are three standard methods for computing finite differences along each axis.
We denote the forward, backward, and central differences that operate along the $x$ axis~(and similarly for the $y$ and $z$ axes) as $\mathcal{D}^{+x}$, $\mathcal{D}^{-x}$, and $\mathcal{D}^{0x}$, respectively.
To approximate $\jacdet$ in 2D or 3D transformations, either the same or a different type of finite difference can be used for different axes.
We denote the central difference approximation of $\jacdet$ in 2D and 3D as $\mathcal{D}^{0x} \mathcal{D}^{0y} \jacdet$ and $\mathcal{D}^{0x} \mathcal{D}^{0y} \mathcal{D}^{0z} \jacdet$, respectively.

Despite its current popularity, the central difference approximation of the Jacobian does not always work as expected in the evaluation of the diffeomorphism property of spatial transformations.
For example, Fig.~\ref{f:fails}(a) shows a transformation around center point $\bm{p}$ that is not diffeomorphic despite the fact that $\mathcal{D}^{0x} \mathcal{D}^{0y}\jacdet = 1$.
In fact, the transformation at $\bm{p}$ has no effect on the computation of $\mathcal{D}^{0x}\mathcal{D}^{0y}\jacdet(\bm{p})$,
even if $\bm{p}$ moves outside the field of view.
We call this the \emph{checkerboard problem} because the central difference based Jacobian computations on the transformations of the ``black'' and ``white'' pixels~(of a checkerboard) are independent of each other.  
\begin{figure}[!t]
    \centering
    \includegraphics[scale=1]{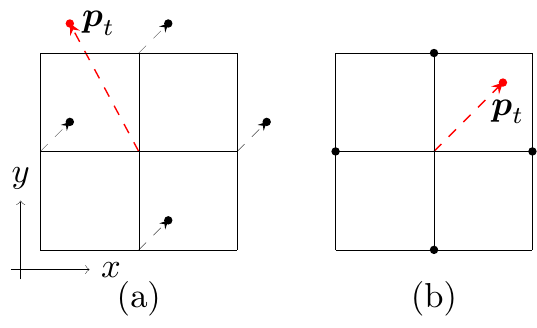}
    \caption{\textbf{(a)}~is an illustration of the \emph{checkerboard problem} for $D^{0x}D^{0y}\jacdet$. \textbf{(b)}~is a transformation that illustrates the inconsistency issue between $D^{-x}D^{-y}\jacdet$ and $D^{+x}D^{+y}\jacdet$.
    In both figures, the center point $\bm{p}$ is transformed to $\bm{p}^{\phantom{y}}_{t}$.
    The transformation around the center point is visualized as displacements~(shown as dotted arrows pointing toward solid dots) and the displacement for the center point is highlighted in red}
    \label{f:fails}
\end{figure}
A possible solution to this problem is to use $\mathcal{D}^{-x}\mathcal{D}^{-y}\jacdet$ or $\mathcal{D}^{+x}\mathcal{D}^{+y}\jacdet$ instead.
However, Fig.~\ref{f:fails}(b) shows an example in which $\mathcal{D}^{-x} \mathcal{D}^{-y} \jacdet$ and $\mathcal{D}^{+x} \mathcal{D}^{+y} \jacdet$ have opposite signs, leading to contradictory conclusions.
These examples illustrate the problem with naive application of finite differences in this application. We seek a better approach that still involves finite differences, but avoids these types of contradictory situations.  

In this work, we first investigate the geometric meaning of finite difference approximations of $\jacdet$.
We show that when using forward or backward differences, the sign of $\jacdet(\bm{p})$ determines if the underlying transformation $\mathcal{T}$ is invertible and orientation-preserving in a triangle~(2D) or tetrahedron~(3D) adjacent to $\bm{p}$.
Reversing the orientation indicates folding in space.
We formally define digital transformations that are globally invertible and free of folding as digital diffeomorphisms.
In order to determine if a transformation is a digital diffeomorphism, at each point it is necessary to consider four finite difference approximated $\jacdet$'s for 2D transformations and ten $\jacdet$'s for 3D transformations.
We also demonstrate that because of the \emph{checkerboard problem} and other errors that are inherent in the central difference based $\jacdet$, the number of non-diffeomorphic voxels it reports is always less than or equal to the actual number. 
Finally, we propose to use non-diffeomorphic area (2D) and volume (3D) as more meaningful measurements of irregularity in computed transformations.
\begin{figure*}[!t]
    \centering
    \includegraphics[scale=1]{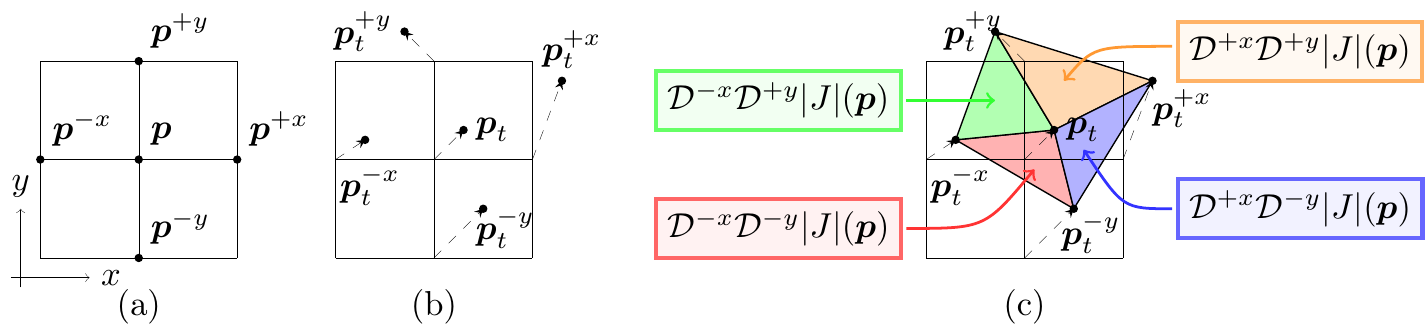}
    \caption{\textbf{(a)}~shows the notations for $\bm{p}$ and its 4-connected neighbors. \textbf{(b)}~shows $\bm{p}^{\phantom{y}}_{t}$, the transformed version of $\bm{p}$ as well as the transformed position of the 4-connected neighbors. \textbf{(c)}~shows the triangular regions and their corresponding forward and backward difference based $\jacdet$'s\label{f:notation_2d}}
\end{figure*}

\section{Methodology}
\label{s:methodology}
\subsection{Backward Difference Based Jacobian Determinant in 2D}
Consider the standard Euclidean space $\mathbb{R}^{2}$ that follows the right-hand rule.
Let $\mathcal{T}$ be a digital transformation for $\mathbb{R}^{2}$ that is defined for every grid point $\bm{p}$.
When using backward differences on both $x$ and $y$ axes for approximating $\jacdet$ at point $\bm{p}$, we have the formulation
\begin{equation}
\mathcal{D}^{-x}\mathcal{D}^{-y}\jacdet(\bm{p}) =
    \left\vert \begin{matrix}
    \mathcal{D}^{-x}\mathcal{T}_{x}(\bm{p})
    &
    \mathcal{D}^{-y}\mathcal{T}_{x}(\bm{p}) \\
    \mathcal{D}^{-x}\mathcal{T}_{y}(\bm{p})
    &
    \mathcal{D}^{-y}\mathcal{T}_{y}(\bm{p}) \\
\end{matrix}\right\vert,
\label{e:backward_diff}
\end{equation}
where $\mathcal{T}_{x}(\bm{p})$ and $\mathcal{T}_{y}(\bm{p})$ are the $x$ and $y$ components of  $\mathcal{T}(\bm{p})$.

We denote the 4-connected neighbors of $\bm{p}$ as $\bm{p}^{-x}, \bm{p}^{+x}, \bm{p}^{-y}$, and $\bm{p}^{+y}$, as shown in Fig.~\ref{f:notation_2d}(a). Their transformed locations are denoted with subscripts $t$, as shown in Fig.~\ref{f:notation_2d}(b). For example, $\bm{p}^{+x}_{t}:=\mathcal{T}(\bm{p}^{+x})$.
We denote the triangular region defined by the vectors $\directline{\bm{p}\bm{p}^{-x}}$ and $\directline{\bm{p}\bm{p}^{-y}}$ as $\triangle \bm{p}\bm{p}^{-x}\bm{p}^{-y}$, and we assume that the 2D transformation $\mathcal{T}$ is linearly interpolated on $\triangle \bm{p}\bm{p}^{-x}\bm{p}^{-y}$.
\begin{proposition}
    A 2D transformation $\mathcal{T}$ is invertible for $\triangle \bm{p}\bm{p}^{-x}\bm{p}^{-y}$ if and only if $\mathcal{D}^{-x}\mathcal{D}^{-y}\left\vert J\right\vert(\bm{p})$ for $\mathcal{T}$ is nonzero.
    \label{p:invertible_2d}
\end{proposition}
\begin{proof}
    Since $\mathcal{T}$ is linearly interpolated on $\triangle \bm{p}\bm{p}^{-x}\bm{p}^{-y}$, $\mathcal{T}$ is linear for $\triangle \bm{p}\bm{p}^{-x}\bm{p}^{-y}$ and can be written as a $2\times2$ matrix with $\directline{\bm{p}^{\phantom{y}}_{t}\bm{p}^{-x}_{t}}$ and $\directline{\bm{p}^{\phantom{y}}_{t}\bm{p}^{-y}_{t}}$ as its columns. Thus, $\mathcal{T}$ is invertible if and only if $\directline{\bm{p}^{\phantom{y}}_{t}\bm{p}^{-x}_{t}}$ and $\directline{\bm{p}^{\phantom{y}}_{t}\bm{p}^{-y}_{t}}$ are linearly independent~(not colinear).
    $\mathcal{D}^{-x}\mathcal{D}^{-y}\jacdet(\bm{p})$ can be written as a triple product:
    \begin{equation}
        \mathcal{D}^{-x}\mathcal{D}^{-y}\jacdet(\bm{p})
        =
        (\directline{\bm{p}^{\phantom{y}}_{t}\bm{p}^{-x}_{t}}\times\directline{\bm{p}^{\phantom{y}}_{t}\bm{p}^{-y}_{t}})\cdot \bm{n},
        \label{e:signed_area_triangle}
    \end{equation}
    where $\bm{n}$ is the unit vector perpendicular to vectors $\directline{\bm{p}^{\phantom{y}}_{t}\bm{p}^{-x}_{t}}$ and $\directline{\bm{p}^{\phantom{y}}_{t}\bm{p}^{-y}_{t}}$ and is positively oriented following the right-hand rule. Therefore, $\mathcal{T}$ is invertible if and only if $\mathcal{D}^{-x} \mathcal{D}^{-y} \jacdet (\bm{p}) \neq 0$.\QED
\end{proof}
\begin{definition}
    A 2D transformation $\mathcal{T}$ is said to \textit{cause folding} of $\triangle \bm{p}\bm{p}^{-x}\bm{p}^{-y}$
    if the orientation of $\triangle \bm{p}\bm{p}^{-x}\bm{p}^{-y}$ is reversed by $\mathcal{T}$.
\end{definition}
\begin{proposition}
    A 2D transformation $\mathcal{T}$ is free of folding for $\triangle\bm{p}\bm{p}^{-x}\bm{p}^{-y}$ if and only if $\mathcal{D}^{-x}\mathcal{D}^{-y}\left\vert J\right\vert(\bm{p})$ for $\mathcal{T}$ is positive.
    \label{p:free_of_folding_2d}
\end{proposition}
\begin{proof}By the right-hand rule, $\triangle\bm{p}\bm{p}^{-x}\bm{p}^{-y}$ is positively oriented.

    \noindent$(\Rightarrow)$ When $\mathcal{T}$ is free of folding, the orientation of $\triangle\bm{p}\bm{p}^{-x}\bm{p}^{-y}$
    is preserved by $\mathcal{T}$.
    Because of linear interpolation, $\triangle \bm{p}\bm{p}^{-x}\bm{p}^{-y}$ is transformed to $\triangle \bm{p}^{\phantom{y}}_{t}\bm{p}^{-x}_{t}\bm{p}^{-y}_{t}$, which is also positively oriented.
    Equation~\ref{e:signed_area_triangle} shows that  $\mathcal{D}^{-x}\mathcal{D}^{-y}\jacdet(\bm{p})$ equals twice the signed area of $\triangle \bm{p}^{\phantom{y}}_{t}\bm{p}^{-x}_{t}\bm{p}^{-y}_{t}$.
    Therefore, $\mathcal{D}^{-x}\mathcal{D}^{-y}\jacdet(\bm{p})>0$.

    \vspace*{1em}
    \noindent$(\Leftarrow)$ When $\mathcal{D}^{-x}\mathcal{D}^{-y}\jacdet(\bm{p})>0$, from Eq.~\ref{e:signed_area_triangle} $\triangle \bm{p}^{\phantom{y}}_{t}\bm{p}^{-x}_{t}\bm{p}^{-y}_{t}$ is positively oriented.
    Because of linear interpolation, $\triangle \bm{p}\bm{p}^{-x}\bm{p}^{-y}$ is transformed to
    $\triangle \bm{p}^{\phantom{y}}_{t}\bm{p}^{-x}_{t}\bm{p}^{-y}_{t}$ and both of them are positively oriented. Therefore, $\mathcal{T}$ is free of folding for $\triangle \bm{p}\bm{p}^{-x}\bm{p}^{-y}$.\QED
\end{proof}
\begin{definition}
    A 2D transformation $\mathcal{T}$ is \textit{digitally diffeomorphic} for the region $\triangle \bm{p}\bm{p}^{-x}\bm{p}^{-y}$ if $\mathcal{T}$ is invertible and free of folding for $\triangle \bm{p}\bm{p}^{-x}\bm{p}^{-y}$.
\end{definition}
\begin{proposition}
    A 2D transformation $\mathcal{T}$ is digitally diffeomorphic for the region $\triangle \bm{p}\bm{p}^{-x}\bm{p}^{-y}$ if and only if  $\mathcal{D}^{-x}\mathcal{D}^{-y}\left\vert J\right\vert(\bm{p})>0$.
    \label{p:digital_diffeomorphism_2d}
\end{proposition}
\begin{proof}
    Proposition~\ref{p:digital_diffeomorphism_2d} is a direct consequence of Propositions~\ref{p:invertible_2d} and~\ref{p:free_of_folding_2d}.\QED
\end{proof}
In conclusion, $\mathcal{D}^{-x}\mathcal{D}^{-y}\left\vert J\right\vert$ for the point $\bm{p}$ informs us about the digitally diffeomorphic property of a triangle adjacent to $\bm{p}$, under the assumption that the transformation is linearly interpolated.
\begin{figure*}[!t]
    \centering
    \includegraphics[scale=1]{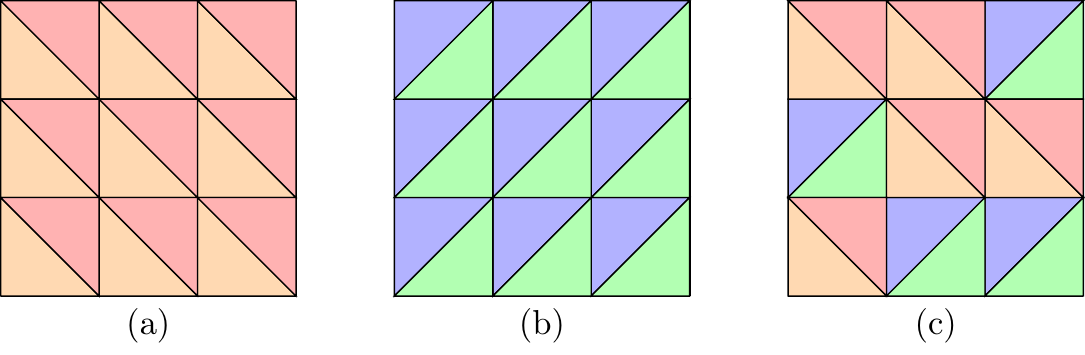}
    \caption{Illustration of the combinations of Jacobian determinants and their corresponding triangular regions. \textbf{(a)}~$\mathcal{D}^{-x}\mathcal{D}^{-y}\jacdet(\bm{p})$ and $\mathcal{D}^{+x}\mathcal{D}^{+y}\jacdet(\bm{p})$. \textbf{(b)}~$\mathcal{D}^{-x}\mathcal{D}^{+y}\jacdet(\bm{p})$ and $\mathcal{D}^{+x}\mathcal{D}^{-y}\jacdet(\bm{p})$. \textbf{(c)}~Each pixel uses a different combination of Jacobian determinants\label{f:meshes}}
\end{figure*}
\subsection{Digital Diffeomorphism in Two Dimensions}
Similar to $\mathcal{D}^{-x}\mathcal{D}^{-y}\jacdet(\bm{p})$, we can replicate Proposition~\ref{p:digital_diffeomorphism_2d} to establish that any $\jacdet(\bm{p})$ approximated using any combination of forward and backward differences is testing if the transformation is digitally diffeomorphic for a triangular region around $\bm{p}$~(see Fig.~\ref{f:notation_2d}(c)), assuming that the region is linearly interpolated.
Since these $\jacdet(\bm{p})$'s cover different regions, their signs are independent of each other.
For example, $\mathcal{T}$ can have a positive $\mathcal{D}^{-x}\mathcal{D}^{-y}\jacdet(\bm{p})$ and a negative $\mathcal{D}^{+x}\mathcal{D}^{+y}\jacdet(\bm{p})$ at the
same time~(see Fig.~\ref{f:fails}(b)).
For each of these approximations, the implied triangular regions for all $\bm{p}$'s taken together cover only half of the space~(\eg $\mathcal{D}^{-x}\mathcal{D}^{-y}\jacdet$ only considers the red tiles in Fig.~\ref{f:meshes}(a)).
Consequently, even if a transformation $\mathcal{T}$ has
$\mathcal{D}^{-x}\mathcal{D}^{-y}\jacdet(\bm{p})>0$ for all $\bm{p}$'s, half the space is not considered and can potentially exhibit folding or be non-invertible.
This is also the case for any other forward and backward difference computations of $\jacdet$.

Although combining the triangular regions of $\mathcal{D}^{-x}\mathcal{D}^{-y}\jacdet$ and $\mathcal{D}^{+x}\mathcal{D}^{+y}\jacdet$ can cover the entire space, positive $\mathcal{D}^{-x}\mathcal{D}^{-y}\jacdet$ and $\mathcal{D}^{+x}\mathcal{D}^{+y}\jacdet$ only guarantee that the transformation is digitally diffeomorphic when the transformation is piecewise linearly interpolated as shown in Fig.~\ref{f:meshes}(a). A different choice, for example that shown in Fig.~\ref{f:meshes}(b), which corresponds to $\mathcal{D}^{-x}\mathcal{D}^{+y}\jacdet$ and $\mathcal{D}^{+x}\mathcal{D}^{-y}\jacdet$ can give contradictory conclusions.
Since there are two ways of dividing the square-size area between grid points, there are many more piecewise linear transformations that correspond to the same digital transformation, \eg Fig.~\ref{f:meshes}(c). To anticipate all possible choices of triangulation---each of which leads to its own invertible transformation on the plane---we should therefore consider all finite difference approximations in determining whether a transformation is digitally diffeomorphic. This leads naturally to the following definition: 
\begin{definition}
    A 2D digital transformation $\mathcal{T}$ is a \textit{digital diffeomorphism} if for every grid point $\bm{p}$ it satisfies
    $\mathcal{D}^{-x}\mathcal{D}^{-y}\jacdet(\bm{p})>0$, $\mathcal{D}^{-x}\mathcal{D}^{+y}\jacdet(\bm{p})>0$, $\mathcal{D}^{+x}\mathcal{D}^{-y}\jacdet(\bm{p})>0$, and $\mathcal{D}^{+x}\mathcal{D}^{+y}\jacdet(\bm{p})>0$.
    \label{def:dd2d}
\end{definition}
The proposed digital diffeomorphism definition guarantees the transformation to be free of folding and invertible regardless of the piecewise linear transformation (on triangles that divide the squares between pixel centers) that is used.
\begin{figure*}[!t]
    \centering
    \includegraphics[scale=1]{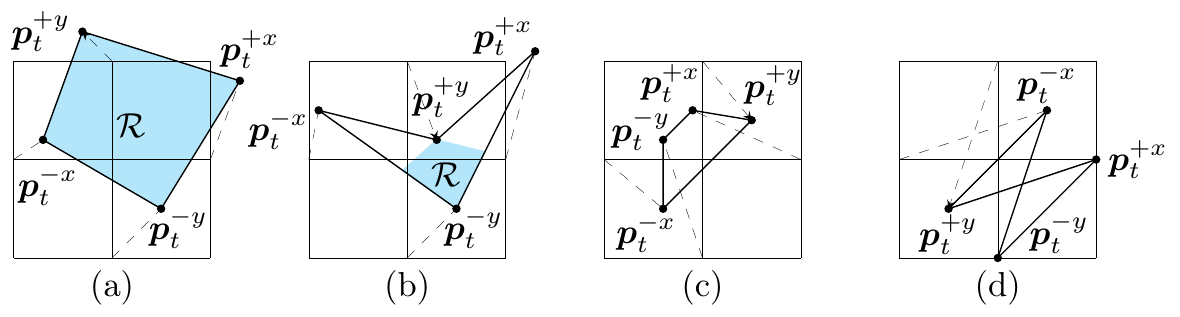}
    \caption{Examples of $\mathcal{R}(\bm{p})$ (blue shaded region) for $\bm{p}^{\phantom{y}}_{t}$ such that
    $\mathcal{T}$ is digitally diffeomorphic. \textbf{(c)}~and \textbf{(d)}~show examples of
    $\mathcal{T}$'s that are causing folds in space and thus $\mathcal{R}(\bm{p})$ is an
    empty set\label{f:polygon}}
\end{figure*}
\subsection{Central Difference Based Jacobian Determinant}

In this section, we analyze the central difference approximation of $\jacdet$ given the previous analysis.  We first ask where $\bm{p}$ can be positioned to yield a digital diffeomorphism when its neighbors have fixed transformations.  We start with the following definitions: 
\begin{definition}
    For grid point $\bm{p}$ with fixed $\bm{p}_t^{-x}$, $\bm{p}_t^{-y}$, $\bm{p}_t^{+x}$, and $\bm{p}_t^{+y}$, $\mathcal{R}(\bm{p})$ is defined to be the region in $\mathbb{R}^{2}$ such that 
    $\mathcal{D}^{-x}\mathcal{D}^{-y}\jacdet(\bm{p})>0$, $\mathcal{D}^{-x}\mathcal{D}^{+y}\jacdet(\bm{p})>0$, $\mathcal{D}^{+x}\mathcal{D}^{-y}\jacdet(\bm{p})>0$, and $\mathcal{D}^{+x}\mathcal{D}^{+y}\jacdet(\bm{p})>0$.  \label{def:r}
\end{definition}

\begin{definition}
    Let $\bm{a}$ and $\bm{b}$ be points in $\mathbb{R}^{2}$. The half-plane $\mathcal{H}(\bm{a},\bm{b})$ is defined as $\bm{p} \in \mathbb{R}^2$ such that $\triangle \bm{a}\bm{b}\bm{p}$ is positively oriented.
\end{definition}
\begin{proposition}
    $\mathcal{R}(\bm{p})$ is the intersection of the four half-planes $\mathcal{H}(\bm{p}^{-x}_{t},\bm{p}^{-y}_{t})$, $\mathcal{H}(\bm{p}^{-y}_{t},\bm{p}^{+x}_{t})$, $\mathcal{H}(\bm{p}^{+x}_{t}\bm{p}^{+y}_{t})$, and $\mathcal{H}(\bm{p}^{+y}_{t},\bm{p}^{-x}_{t})$.
    \label{p:region_r}
\end{proposition}
\begin{proof}
    Equation~\ref{e:backward_diff} can be written as:
    $$
        \mathcal{D}^{-x}\mathcal{D}^{-y}\jacdet(\bm{p})
        =
        (\directline{\bm{p}^{-x}_{t}\bm{p}^{-y}_{t}}\times\directline{\bm{p}^{-x}_{t}\bm{p}^{\phantom{y}}_{t}})\cdot \bm{n}.
    $$
    Therefore, when $\mathcal{D}^{-x}\mathcal{D}^{-y}\jacdet(\bm{p})>0$, $\triangle\bm{p}^{-x}_{t}\bm{p}^{-y}_{t}\bm{p}^{\phantom{y}}_{t}$ is positively oriented. Thus $\bm{p}^{\phantom{y}}_{t}$ is in the half-plane $\mathcal{H}(\bm{p}^{-x}_{t},\bm{p}^{-y}_{t})$.
    Analogous statements can be made for the other forward and backward difference based $\jacdet$'s.
    When all the four $\jacdet$'s are positive, $\bm{p}^{\phantom{y}}_{t}$ must be inside the four half-planes $\mathcal{H}(\bm{p}^{-x}_{t},\bm{p}^{-y}_{t})$, $\mathcal{H}(\bm{p}^{-y}_{t},\bm{p}^{+x}_{t})$, $\mathcal{H}(\bm{p}^{+x}_{t}\bm{p}^{+y}_{t})$, and $\mathcal{H}(\bm{p}^{+y}_{t},\bm{p}^{-x}_{t})$. Therefore, $\mathcal{R}(\bm{p})$ is the intersection of these four half-planes.\QED
\end{proof}
Figure~\ref{f:polygon} shows examples of $\mathcal{R}(\bm{p})$.
In Figs.~\ref{f:polygon}(c) and~(d), no matter where $\bm{p}^{\phantom{y}}_{t}$ is located, at least one forward or backward difference based $\jacdet(\bm{p})$ is guaranteed to be negative and, thus, $\mathcal{R}(\bm{p})$ is empty. This makes sense since the transformations
in Figs.~\ref{f:polygon}(c) and~(d) cause folds in space.
\begin{proposition}
    Assume $\bm{p}^{-x}_{t}$, $\bm{p}^{-y}_{t}$, $\bm{p}^{+x}_{t}$, and $\bm{p}^{+x}_{t}$ forms a simple polygon~(without self-intersection).
    Then $\mathcal{R}(\bm{p})$ is non-empty if and only if $\mathcal{D}^{0x}\mathcal{D}^{0y}\left\vert J\right\vert(\bm{p})>0$.
    \label{p:central_difference_2d}
\end{proposition}
\begin{proof}
    $\mathcal{D}^{0x}\mathcal{D}^{0y}\jacdet(\bm{p})$ can be written as:
    \begin{align}
        \mathcal{D}^{0x}\mathcal{D}^{0y}\jacdet
        =\frac{1}{2}
        &
        \left(\mathcal{D}^{-x}\mathcal{D}^{-y}\jacdet+\mathcal{D}^{+x}\mathcal{D}^{-y}\jacdet \right.\nonumber
        \\
        &
        \left.+\mathcal{D}^{-x}\mathcal{D}^{+y}\jacdet+\mathcal{D}^{+x}\mathcal{D}^{+y}\jacdet\right).
        \label{e:area_sum}
    \end{align}

    \noindent$(\Rightarrow)$ If $\mathcal{R}(\bm{p})$ is non-empty, for every $\bm{p}^{\phantom{y}}_{t}\in\mathcal{R}(\bm{p})$ all its forward and backward difference based $\jacdet(\bm{p})$'s are positive by Definition~\ref{def:r}, and therefore from Eq.~\ref{e:area_sum}, $\mathcal{D}^{0x}\mathcal{D}^{0y}\jacdet(\bm{p})>0$.

    \vspace*{1em}
    \noindent$(\Leftarrow)$ On the other hand, if $\mathcal{D}^{0x}\mathcal{D}^{0y}\jacdet>0$, the polygon $\bm{p}^{-x}_{t}\bm{p}^{-y}_{t}\bm{p}^{+x}_{t}\bm{p}^{+y}_{t}$ is positively oriented~\citep{braden1986surveyor}~(\ie interior to the left).
    Thus, if there exists a $\bm{p}^{\phantom{y}}_{t}$ that is visible to all vertices of the polygon, then $\bm{p}^{\phantom{y}}_{t}\in\mathcal{R}(\bm{p})$ by Proposition~\ref{p:region_r}.
    Following~\citep{chvatal1975combinatorial}, such $\bm{p}^{\phantom{y}}_{t}$ always exists for simple polygons with five or fewer vertices.
    Therefore, $\mathcal{R}(\bm{p})$ is non-empty.\QED
\end{proof}

Although a positive $\mathcal{D}^{0x}\mathcal{D}^{0y}\jacdet(\bm{p})$ ensures that $\mathcal{R}(\bm{p})$ is non-empty, $\bm{p}$ can still be digitally non-diffeomorphic when $\bm{p}^{\phantom{y}}_{t}\not\in\mathcal{R}(\bm{p})$, which we see in the \emph{checkerboard problem}.
In addition to the \emph{checkerboard problem}, $\mathcal{D}^{0x}\mathcal{D}^{0y}\jacdet$ also fails to provide a meaningful interpretation when the polygon $\bm{p}^{-x}_{t}\bm{p}^{-y}_{t}\bm{p}^{+x}_{t}\bm{p}^{+y}_{t}$ exhibits self-intersection~(see Fig.~\ref{f:polygon}(d)).
In conclusion, $\mathcal{D}^{0x}\mathcal{D}^{0y}\jacdet(\bm{p}) \leq 0$ always indicates that $\mathcal{T}$ is digitally non-diffeomorphic but $\mathcal{D}^{0x}\mathcal{D}^{0y}\jacdet(\bm{p}) > 0$ does not mean it is digitally diffeomorphic because of the checkerboard or self-intersection problems. Therefore, use of central differences alone to estimate $\jacdet$  is an inadequate characterization of digital diffeomorphism.

\subsection{Digital Diffeomorphism in Three Dimensions}
\label{s:3d}
Consider the standard Euclidean space $\mathbb{R}^{3}$ that follows the right-hand rule.
Let $\mathcal{T}$ be a digital transformation of $\mathbb{R}^{3}$ that is defined for every grid point $\bm{p}$.
We denote the 6-connected neighbors of $\bm{p}$ as $\bm{p}^{\pm x}, \bm{p}^{\pm y}, \bm{p}^{\pm z}$ (see Fig.~\ref{f:demonstration_3d}(a)) and their transformed locations are denoted with the subscript $t$.
We denote the tetrahedron defined by the vectors $\directline{\bm{p}\bm{p}^{-x}}$, $\directline{\bm{p}\bm{p}^{-y}}$, and $\directline{\bm{p}\bm{p}^{-z}}$ as $\tetrahedron \bm{p}\bm{p}^{-x}\bm{p}^{-y}\bm{p}^{-z}$, and we
assume that the 3D transformation $\mathcal{T}$ is linearly interpolated on $\tetrahedron \bm{p}\bm{p}^{-x}\bm{p}^{-y}\bm{p}^{-z}$.
\begin{proposition}
    A 3D transformation $\mathcal{T}$ is invertible for $\tetrahedron \bm{p}\bm{p}^{-x}\bm{p}^{-y}\bm{p}^{-z}$ if and only if $\mathcal{D}^{-x} \mathcal{D}^{-y} \mathcal{D}^{-z}\jacdet$ for $\bm{p}$ is nonzero.
    \label{p:invertible_3d}
\end{proposition}
\begin{proof}
    Since $\mathcal{T}$ is linearly interpolated on $\tetrahedron \bm{p}\bm{p}^{-x}\bm{p}^{-y}\bm{p}^{-z}$, $\mathcal{T}$ is linear for 
    $\tetrahedron \bm{p}\bm{p}^{-x}\bm{p}^{-y}\bm{p}^{-z}$ and can be written as a $3\times3$ matrix with $\directline{\bm{p}^{\phantom{y}}_{t}\bm{p}^{-x}_{t}}$, $\directline{\bm{p}^{\phantom{y}}_{t}\bm{p}^{-y}_{t}}$, and $\directline{\bm{p}^{\phantom{y}}_{t}\bm{p}^{-z}_{t}}$ as its columns.
    Thus, $\mathcal{T}$ is invertible if and only if $\directline{\bm{p}^{\phantom{y}}_{t}\bm{p}^{-x}_{t}}$, $\directline{\bm{p}^{\phantom{y}}_{t}\bm{p}^{-y}_{t}}$, and $\directline{\bm{p}^{\phantom{y}}_{t}\bm{p}^{-z}_{t}}$ are linearly independent~(not colinear).
    $\mathcal{D}^{-x} \mathcal{D}^{-y} \mathcal{D}^{-z}\jacdet(\bm{p})$ can be written as a triple product:
    \begin{equation}
        \mathcal{D}^{-x}\mathcal{D}^{-y}\mathcal{D}^{-z}\jacdet(\bm{p}) =
        -(\directline{\bm{p}^{\phantom{y}}_{t}\bm{p}^{-x}_{t}}\times\directline{\bm{p}^{\phantom{y}}_{t}\bm{p}^{-y}_{t}})\cdot\directline{\bm{p}^{\phantom{y}}_{t}\bm{p}^{-z}_{t}}.
        \label{e:backward_diff_3d}
    \end{equation}
    Therefore, $\mathcal{T}$ is invertible if and only if $\mathcal{D}^{-x} \mathcal{D}^{-y} \mathcal{D}^{-z}\jacdet(\bm{p}) \neq 0$.\QED
\end{proof}
\begin{figure*}[!t]
    \centering
    \includegraphics[scale=1]{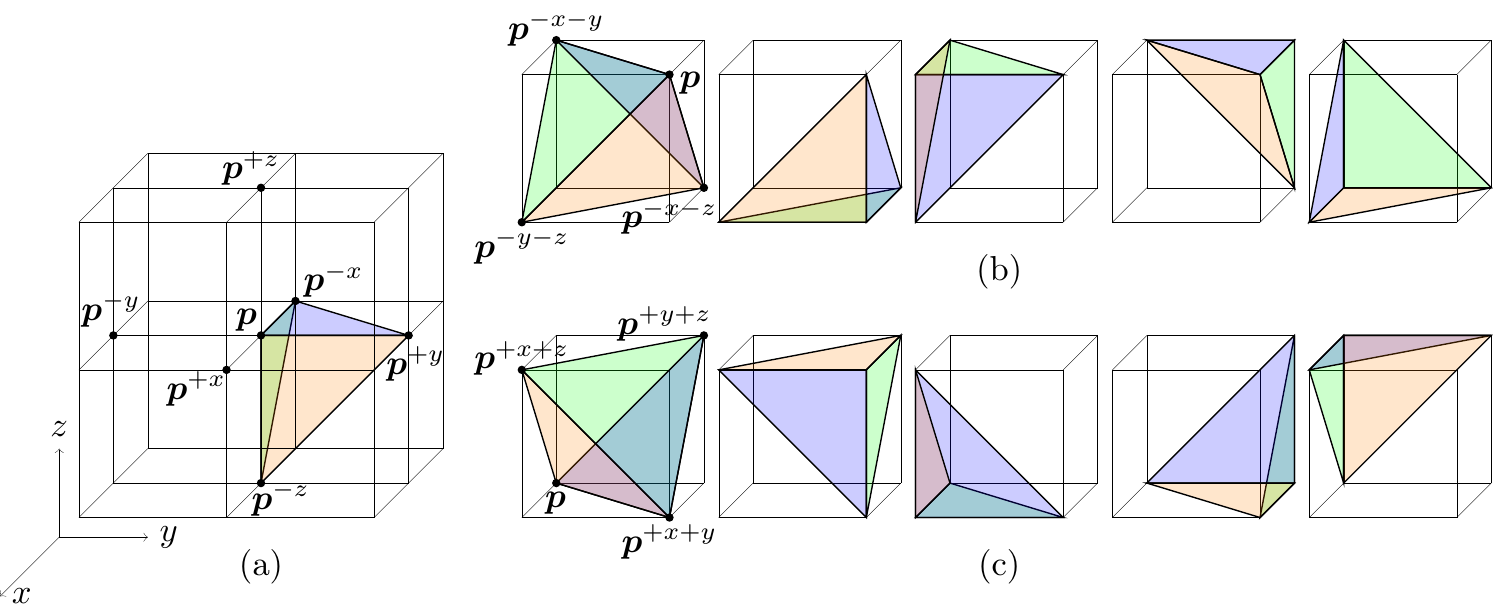}
    \caption{\textbf{(a)}~shows the notations for $\bm{p}$ and its 6-connected neighbors in 3D and the tetrahedron considered by $\mathcal{D}^{-x}\mathcal{D}^{+y}\mathcal{D}^{-z}\jacdet$. 
    \textbf{(b)}~and \textbf{(c)}~are illustrations of the two schemes to divide the cube volume in-between grid points in 3D 
    \label{f:demonstration_3d}}
\end{figure*}
\begin{definition}
    A 3D transformation $\mathcal{T}$ is said to cause \textit{folding} for $\tetrahedron \bm{p}\bm{p}^{-x}\bm{p}^{-y}\bm{p}^{-z}$ if the orientation of $\tetrahedron \bm{p}\bm{p}^{-x}\bm{p}^{-y}\bm{p}^{-z}$ is reversed by $\mathcal{T}$.
\end{definition}
\begin{proposition}
    A 3D transformation $\mathcal{T}$ is free of folding for $\tetrahedron \bm{p}\bm{p}^{-x}\bm{p}^{-y}\bm{p}^{-z}$ if and only if $\mathcal{T}$ has $\mathcal{D}^{-x} \mathcal{D}^{-y} \mathcal{D}^{-z}\jacdet(\bm{p}) > 0$.
    \label{p:free_of_folding_3d}
\end{proposition}
\begin{proof}
    $(\Rightarrow)$ When $\mathcal{T}$ is free of folding, the orientation of $\tetrahedron \bm{p}\bm{p}^{-x}\bm{p}^{-y}\bm{p}^{-z}$~(negatively oriented by the right-hand rule) is preserved by $\mathcal{T}$.
    Because of linear interpolation, $\tetrahedron \bm{p}\bm{p}^{-x}\bm{p}^{-y}\bm{p}^{-z}$ is transformed to $\tetrahedron \bm{p}^{\phantom{y}}_{t}\bm{p}_{t}^{-x}\bm{p}_{t}^{-y}\bm{p}_{t}^{-z}$, which is also negatively oriented. 
    Equation~\ref{e:backward_diff_3d} shows that $\mathcal{D}^{-x}\mathcal{D}^{-y}\mathcal{D}^{-z}\jacdet(\bm{p})$ equals to six times the negative signed volume of $\tetrahedron \bm{p}^{\phantom{y}}_{t}\bm{p}_{t}^{-x}\bm{p}_{t}^{-y}\bm{p}_{t}^{-z}$.
    Therefore, $\mathcal{D}^{-x}\mathcal{D}^{-y}\mathcal{D}^{-z}\jacdet(\bm{p})>0$.

    \vspace*{1em}
    \noindent$(\Leftarrow)$ $\mathcal{D}^{-x}\mathcal{D}^{-y}\mathcal{D}^{-z}\jacdet(\bm{p})>0$ indicates that $\tetrahedron \bm{p}^{\phantom{y}}_{t}\bm{p}_{t}^{-x}\bm{p}_{t}^{-y}\bm{p}_{t}^{-z}$ is negatively oriented by Eq.~\ref{e:backward_diff_3d}.
    Because of linear interpolation $\tetrahedron \bm{p}\bm{p}^{-x}\bm{p}^{-y}\bm{p}^{-z}$ is mapped to $\tetrahedron \bm{p}^{\phantom{y}}_{t}\bm{p}_{t}^{-x}\bm{p}_{t}^{-y}\bm{p}_{t}^{-z}$ and both of them are negatively oriented.
    Therefore, $\mathcal{T}$ is free of folding for $\tetrahedron \bm{p}\bm{p}^{-x}\bm{p}^{-y}\bm{p}^{-z}$.\QED
\end{proof}
\begin{definition}
    A 3D transformation $\mathcal{T}$ is \textit{digitally diffeomorphic} for
    the region $\tetrahedron \bm{p}\bm{p}^{-x}_{}\bm{p}^{-y}_{}\bm{p}^{-z}_{}$ if $\mathcal{T}$ is invertible and free of folding for $\tetrahedron \bm{p}\bm{p}^{-x}\bm{p}^{-y}\bm{p}^{-z}$.
\end{definition}
\begin{proposition}
    A 3D transformation $\mathcal{T}$ is digitally diffeomorphic for the region $\tetrahedron \bm{p}\bm{p}^{-x}_{}\bm{p}^{-y}_{}\bm{p}^{-z}_{}$ if and only if $\mathcal{D}^{-x}\mathcal{D}^{-y}\mathcal{D}^{-z}\jacdet(\bm{p})>0$.
    \label{p:digital_diffeomorphism_3d}
\end{proposition}
\begin{proof}
    Proposition~\ref{p:digital_diffeomorphism_3d} is a direct consequence of
    Propositions~\ref{p:invertible_3d} and~\ref{p:free_of_folding_3d}.\QED
\end{proof}
Similarly, we can prove that a $\jacdet(\bm{p})$ approximated using any combination of forward and backward differences is testing if the transformation is digitally diffeomorphic in a tetrahedron adjacent to $\bm{p}$.
One of these tetrahedra~(for $\mathcal{D}^{-x}\mathcal{D}^{+y}\mathcal{D}^{-z}\jacdet$) is visualized in Fig.~\ref{f:demonstration_3d}(a). 
However, unlike the 2D case where all the triangular regions completely cover the entire 2D space, the union of all these adjacent tetrahedra does not fill the entire 3D space.
Therefore, even if all forward and backward difference approximations of $\jacdet$ are positive, the transformation can still cause folding in the spaces not covered by these adjacent tetrahedra. 

To solve this issue, we introduce another tetrahedron $\tetrahedron\bm{p}\bm{p}^{-x-y}_{}\bm{p}^{-x-z}_{}\bm{p}^{-y-z}_{}$, as shown in Fig.~\ref{f:demonstration_3d}(b).
When combined with four existing tetrahedra from finite difference based $\jacdet$'s, they completely cover the
entire volume.
We define $\vert J^\star_1\vert(\bm{p})$ as
\begin{equation}
    \vert J^{\star}_{1}\vert(\bm{p}) = 
    (\directline{\bm{p}^{\phantom{y}}_{t}\bm{p}^{-x-y}_{t}}\times\directline{\bm{p}^{\phantom{y}}_{t}\bm{p}^{-x-z}_{t}})\cdot\directline{\bm{p}^{\phantom{y}}_{t}\bm{p}^{-y-z}_{t}},
    \label{e:extra_tetrahedron_1}
\end{equation}
where $\bm{p}^{-x-y}_{t}$, $\bm{p}^{-x-z}_{t}$, and $\bm{p}^{-y-z}_{t}$ are the transformed locations of $\bm{p}^{-x-y}$, $\bm{p}^{-x-z}$, $\bm{p}^{-y-z}$.
The signed volume of the extra tetrahedron after applying $\mathcal{T}$ equals $\frac{1}{6}\vert J^{\star}_{1}\vert(\bm{p})$.
Similar to 2D, there are two ways of dividing the cube-size volume in-between grid points into five tetrahedra, as shown in Figs.~\ref{f:demonstration_3d}(b) and~(c).
The signed volume of the extra tetrahedron in Fig.~\ref{f:demonstration_3d}(c) can be computed as $\frac{1}{6}\vert J^{\star}_{2}\vert (\bm{p})$, where
\begin{equation}
    \vert J^{\star}_{2}\vert(\bm{p}) = 
    (\directline{\bm{p}^{\phantom{y}}_{t}\bm{p}^{+x+y}_{t}}\times\directline{\bm{p}^{\phantom{y}}_{t}\bm{p}^{+y+z}_{t}})\cdot\directline{\bm{p}^{\phantom{y}}_{t}\bm{p}^{+x+z}_{t}}.
    \label{e:extra_tetrahedron_2}
\end{equation}
There are other ways to divide a cube into tetrahedra~\citep{carr2006artifacts}.
These other schemes are not considered here because 1)~their computation involve finite difference approximations at different points than those we consider, or 2)~their computation requires interpolating the transformation at non-grid point.
\begin{definition}
    A 3D digital transformation $\mathcal{T}$ is a digital diffeomorphism if for every grid point $\bm{p}$ its forward and backward difference based $\jacdet$'s of $\bm{p}$ are all positive and both $\vert J^{\star}_{1}\vert (\bm{p})$ and $\vert J^{\star}_{2}\vert (\bm{p})$ are positive.
    \label{def:dd3d}
\end{definition}
For similar reason as in 2D, our definition of digital diffeomorphism involves both schemes shown in Figs.~\ref{f:demonstration_3d}(b) and~(c).
Specifically, each of the two schemes corresponds to a particular piecewise linear transformation.
But for a given digital transformation there are many plausible piecewise linear transformations~(see Fig.~\ref{f:meshes}(c)).
By considering both schemes, we avoid the ambiguity of choosing different schemes for every cube-sized volume.

The central difference approximation of $\jacdet$ in 3D calculates the signed volume of an octahedron with vertices $\bm{p}^{-x}_{t}$, $\bm{p}^{+x}_{t}$, $\bm{p}^{-y}_{t}$, $\bm{p}^{+y}_{t}$, $\bm{p}^{-z}_{t}$, and $\bm{p}^{+z}_{t}$, when the octahedron is simple.
The proof is a straightforward extension of Eq.~\ref{e:area_sum}.
It is easy to show that $\mathcal{D}^{0x}\mathcal{D}^{0y}\mathcal{D}^{0z}\jacdet$ can also have the checkerboard problem or the self-intersection problem as in the 2D case.
Note that Proposition~\ref{p:central_difference_2d} cannot be generalized to 3D because $\bm{p}^{-x}_{t}\bm{p}^{+x}_{t}\bm{p}^{-y}_{t}\bm{p}^{+y}_{t}\bm{p}^{-z}_{t}\bm{p}^{+z}_{t}$ may not be tetrahedralizable~\citep{o1987art} and thus, $\mathcal{D}^{0x}\mathcal{D}^{0y}\mathcal{D}^{0z}\jacdet(\bm{p})>0$ cannot guarantee $\mathcal{R}(\bm{p})$ in 3D is non-empty.
\begin{figure}[!t]
    \centering
    \includegraphics[scale = 0.97]{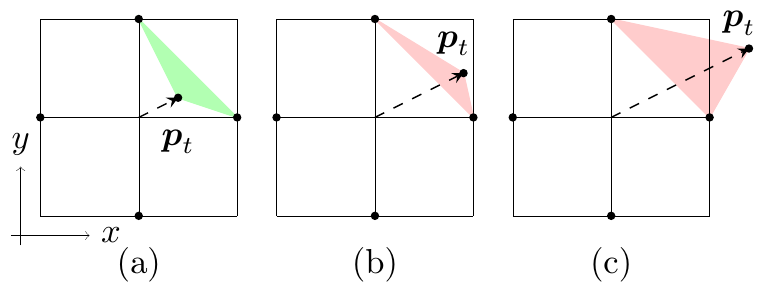}
    \caption{A demonstration of measuring non-diffeomorphic space in 2D. The transformations are visualized as displacement fields. In each of the three subfigures, only the center points are transformed to their corresponding $\bm{p}^{\phantom{y}}_{t}$'s and the other grid points remain fixed. The triangular region corresponds to the forward difference based $\jacdet$ is shown in green if $\jacdet > 0$, otherwise it is shown in red \label{f:nda_demo}}
\end{figure}
\begin{figure*}[!t]
    \centering
    \begin{tabular}{cccc}
        \includegraphics[width=0.22\textwidth]{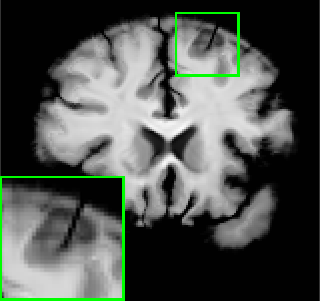}
        &
        \includegraphics[width=0.22\textwidth]{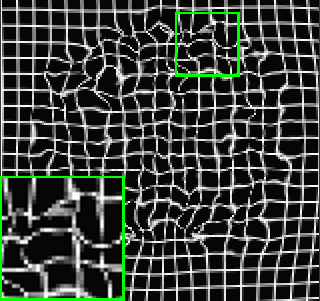}
        &
        \includegraphics[width=0.22\textwidth]{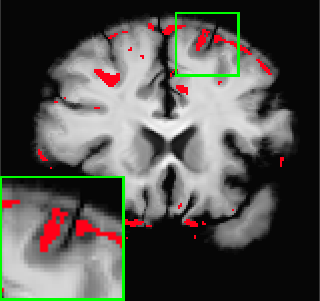}
        &
        \includegraphics[width=0.22\textwidth]{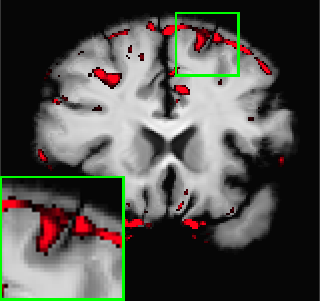}
        \\
        \textbf{(a)} & \textbf{(b)} & \textbf{(c)} & \textbf{(d)}
    \end{tabular}
    \caption{A visualization of the proposed non-diffeomorphic area~(only the displacement inside the coronal plane was considered in this example). \textbf{(a)}~The warped image. \textbf{(b)}~The grid line representation of the transformation~(generated using Voxelmorph~\citep{balakrishnan2019voxelmorph}).
    \textbf{(c)}~The warped image with non-diffeomorphic pixels~(marked in red) as measured by the central difference Jacobian determinant~($\mathcal{D}^{0x}\mathcal{D}^{0y}\jacdet$) highlighted in red. \textbf{(d)}~The warped image with a map overlay indicating the non-diffeomorphic area with brighter shades of red indicating larger non-diffeomorphic area
    \label{f:ixi}}
\end{figure*}

\subsection{Non-diffeomorphic Space Measurement}
For non-diffeomorphic transformations in 3D, the number and percentage of non-diffeomorphic voxels is often used to measure its irregularity.
Specifically, a voxel is considered non-diffeomorphic if its center location $\bm{p}$ has a central difference based Jacobian determinant that is not positive (\ie in 3D: $\mathcal{D}^{0x}\mathcal{D}^{0y} \mathcal{D}^{0z} \jacdet(\bm{p}) \leq 0$).
However, as we have demonstrated in Sections~\ref{s:intro} and \ref{s:methodology}, the central difference based Jacobian determinant underestimates the non-diffeomorphic space, in general.
Given the limitations of the central difference-based Jacobian determinant, we seek an alternative way of evaluating the diffeomorphic property of digital transformations.
We aim to find a quantitative measure that involves the computation of finite difference-based Jacobians and is consistent with our definition of digital diffeomorphism.
As a result, we propose \emph{non-diffeomorphic volume}~(NDV) to quantify the size of the non-diffeomorphic space in 3D caused by a digital transformation.

Given a grid point $\bm{p}$, we denote its eight forward and backward difference based Jacobian determinants as $\vert J_{i}\vert , i\in[1,\ldots,8]$. 
As shown in Section~\ref{s:3d}, each of these determinants is equal to six times the signed volume of a tetrahedron adjacent to $\bm{p}$.
Thus, the non-diffeomorphic volume caused by a given transformation within any tetrahedron is given by
$-\min(\vert J_{i}\vert, 0)/6$,
which is positive only if the tetrahedron is folded.

Following Def.~\ref{def:dd3d}, we use the total volume of all \emph{folded} tetrahedrons from 1)~$\vert J_{i}\vert , i\in[1,\ldots,8]$, 2)~$\vert J^{\star}_{1}\vert (\bm{p})$ given in Eq.~\ref{e:extra_tetrahedron_1}, and 3) $\vert J^{\star}_{2}\vert (\bm{p})$ given in Eq.~\ref{e:extra_tetrahedron_2} to compute NDV exhibited in the entire domain as:
\begin{align}
\text{NDV} & =-\frac{1}{2} \sum_{\bm{p}}
\left[\sum_{i=1}^{8}\frac{ \min (\vert J_{i}\vert (\bm{p}),0)}{6} + \right.
\nonumber\\
&\frac{ \min (\vert J^{\star}_{1}\vert (\bm{p}),0)}{6}
\left. + \frac{ \min (\vert J^{\star}_{2}\vert (\bm{p}),0)}{6}
\right].
\label{e:ndv}
\end{align}
This computation is the average of the two tetrahedralization schemes shown in Figs.~\ref{f:demonstration_3d}(b) and~(c).
While it is possible to report the minimum or maximum achievable non-diffeomorphic volume by tetrahedralizing each cube between voxels based on the specific transformation, doing so would inevitably associate the given digital transformation with one particular piecewise linear transformation.
Therefore, averaging the two tetrahedralization schemes provides a more comprehensive and unbiased measure for a digital transformation. Further discussion on the impact of spatially varying tetrahedralization schemes can be found in Section~\ref{s:dis}.

The proposed NDV is connected to the ideas of simplex counting~(SC)~\citep{yushkevich2010bias} and surface propagation~(SP)~\citep{pai2016deformation}, which are used to assess the degree of volume change on a per-voxel basis.
However, those methods typically use only one scheme to discretize the space.
In contrast, NDV considers the two tetrahedralization schemes shown in Figs.~\ref{f:demonstration_3d}(b) and~(c).
This choice is motivated by the definition of digital diffeomorphism but it also offers a rotation-invariant property.
Specifically, rotating the transformation by any multiple of 90 degrees does not affect its result.
This cannot be achieved using only one scheme.
Further discussion on the relationship and distinction between SC, SP, and the proposed NDV can be found in Section~\ref{s:dis}.

For completeness we note the 2D version of non-diffeomorphic space, which we term \emph{non-diffeomorphic area}~(NDA), follows from Def.~\ref{def:dd2d},
\begin{equation}
    \text{NDA} = -\frac{1}{2} \sum_{\bm{p}}
        \sum_{i=1}^{4} \frac{ \min (\vert J_{i} \vert (\bm{p}),0)}{2},
    \label{e:nda}
\end{equation}
where $\vert J_{i}\vert , i\in[1,\ldots,4]$, are the Jacobian determinants approximated using the four possible combinations of forward and backward differences. 

A demonstration of NDA is provided in Fig.~\ref{f:nda_demo}.
When using the central difference based $\jacdet$, all three cases in Fig.~\ref{f:nda_demo} would be considered diffeomorphic~($\mathcal{D}^{0x} \mathcal{D}^{0y}\mathcal{D}^{0z} \jacdet(\bm{p}) > 0$) because of the checkerboard problem.
The forward difference based $\jacdet$ is able to identify that Figs.~\ref{f:nda_demo}(b) and~(c) exhibit folding, but only NDA can provide the observation that the non-diffeomorphic space caused by the transformation shown in Fig.~\ref{f:nda_demo}(c) is larger than the non-diffeomorphic space in Fig.~\ref{f:nda_demo}(b).

\begin{table*}[!tb]
    \centering
    \caption{Our proposed non-diffeomorphic volume and several other measures on the \texttt{IXI} dataset and the seven comparison algorithms. For `$\mathcal{D}^{0x} \mathcal{D}^{0y} \mathcal{D}^{0z} \jacdet \leq 0$' and `Any $\vert J^{}_{i}\vert \leq0$' we report the mean number of voxels (voxel \#) over the 115 test subjects and the corresponding standard deviation ($\pm$), as well as the percentage (\%) with respect to the brain mask.
    We also report our proposed measure of non-diffeomorphic space---\ie non-diffeomorphic volume~(NDV) in 3D---and the corresponding standard deviations and percentages. Methods are listed in the order in which they were published.
    \label{t:ixi}}
    \adjustbox{max width = 0.99 \textwidth}
    {
        \begin{tabular}{l p{0.25\textwidth} C{0.18\textwidth} C{0.18\textwidth} C{0.18\textwidth}
        }
            \toprule
            && $\mathcal{D}^{0x}\mathcal{D}^{0y}\mathcal{D}^{0z}\jacdet\leq0$
            &
            Any $\vert J^{}_{i}\vert \leq0$
            &
            \cellcolor{green!25}Proposed
            \\
            \cmidrule(lr){3-3} \cmidrule(lr){4-4} \cmidrule(lr){5-5}
            &&
            \# of voxels & \# of voxels & NDV
            \\
            && (\%) & (\%) & (\%)
            \\
            \midrule
            & & $222.2\pm776.1$ & $288.9\pm945.7$ & $10.9\pm45.4$
            \\
            \multirow{-2}{*}{\textit{NiftyReg}} & \multirow{-2}{*}{\citep{modat2010fast}}
            & $(0.01\%)$ & $(0.02\%)$ & $(<0.00\%)$
            \\
            \rowcolor{Gray} & & $5704.7\pm1939.6$ & $12212.9\pm3118.1$ & $1597.4\pm661.5$
            \\
            \rowcolor{Gray}\multirow{-2}{*}{\textit{deedsBCV}} & \multirow{-2}{*}{\citep{heinrich2015multi}}
            & $(0.37\%)$ & $(0.79\%)$ & $(0.10\%)$
            \\
            & & $41233.1\pm8091.3$ & $98241.0\pm17061.1$ & $16261.5\pm2709.3$
            \\
            \multirow{-2}{*}{\textit{Voxelmorph}} & \multirow{-2}{*}{\citep{balakrishnan2019voxelmorph}}
            & $(2.64\%)$ & $(6.26\%)$ & $(1.04\%)$
            \\
            \rowcolor{Gray} & & $44126.5\pm8526.4$ & $99560.0\pm16833.8$ & $17923.6\pm3083.3$
            \\
            \rowcolor{Gray}\multirow{-2}{*}{\textit{Cyclemorph}} & \multirow{-2}{*}{\citep{kim2021cyclemorph}}
            & $(2.83\%)$ & $(6.38\%)$ & $(1.15\%)$
            \\
            & & $0.0\pm0.0$ & $0.1\pm0.7$ & $0.0\pm0.0$
            \\
            \multirow{-2}{*}{\textit{MIDIR}} & \multirow{-2}{*}{\citep{qiu2021learning}}
            & $(<0.00\%)$ & $(<0.00\%)$ & $(<0.00\%)$
            \\
            \rowcolor{Gray} & & $35324.1\pm7887.7$ & $88263.1\pm17261.7$ & $14034.7\pm2902.5$
            \\
            \rowcolor{Gray}\multirow{-2}{*}{\textit{Transmorph}} & \multirow{-2}{*}{\citep{chen2022transmorph}}
            & $(2.26\%)$ & $(5.65\%)$ & $(0.90\%)$
            \\
            & & $8291.4\pm5928.4$ & $20282.7\pm7853.8$ & $822.6\pm248.5$
            \\
            \multirow{-2}{*}{\textit{im2grid}} & \multirow{-2}{*}{\citep{liu2022coordinate}}
            & $(0.53\%)$ &  $(1.3\%)$ & $(0.05\%)$
            \\
        \bottomrule
        \end{tabular}
    }
\end{table*}

\section{Experiments}
\label{s:expt}
We compared the commonly used central difference based Jacobian determinant ($\mathcal{D}^{0x} \mathcal{D}^{0y} \mathcal{D}^{0z} \jacdet$) and our proposed non-diffeomorphic volume~(NDV) using several deformable registration algorithms on two publicly available datasets:
\begin{description}
    \item[\texttt{IXI}:] A total of $576$ T1–weighted brain magnetic resonance~(MR) images from the publicly available IXI dataset were used.
    $403$ scans were used in training for the task of atlas-to-subject registration~\citep{kim2021cyclemorph} and $58$ scans were used for validation.
    The transformations generated from registering an atlas brain MR images to $115$ test scans were evaluated.\\
    
    \item[\texttt{Learn2Reg OASIS}:] We also used the brain T1-weighted MR images from the 2021 Learn2Reg challenge~\citep{hering2021learn2reg,lamontagne2019oasis}.
    Scans were preprocessed using FreeSurfer~\citep{fischl2012freesurfer,hoopes2021hypermorph}. All algorithms were trained using the training set of $414$ scans and the transformations for the $19$ validation pairs were evaluated.
\end{description}

The central difference Jacobian determinant approximation was implemented directly from the 2021 Learn2Reg challenge evaluation script.
The implementation details and hyper-parameters for each of the algorithms were adopted from~\citep{chen2022transmorph} and~\citep{liu2022coordinate}.

A visualization of a result from the \texttt{IXI} dataset is shown in Fig.~\ref{f:ixi}.
Since it is difficult to visualize displacements across slices~(anterior-to posterior direction in this case), only the displacements within the coronal plane were considered.
Thus, we computed the non-diffeomorphic area instead of non-diffeomorphic volume.
For each pixel in Fig.~\ref{f:ixi}(d), a higher red intensity indicates larger non-diffeomorphic area around that pixel. The non-diffeomorphic pixels computed from $\mathcal{D}^{0x}\mathcal{D}^{0y} \jacdet$ are highlighted in Fig.~\ref{f:ixi}(c) for comparison.

The results on the \texttt{IXI} dataset are summarized in Table~\ref{t:ixi} and for the \texttt{Learn2Reg OASIS} in Table~\ref{t:l2r}.
Only the voxels within the brain were considered and the percentages were calculated relative to the brain volume of the fixed image.
In both tables, we report the number of non-diffeomorphic voxels~(\# of voxel) and its percentage~(\%) based on the $\mathcal{D}^{0x}\mathcal{D}^{0y}\mathcal{D}^{0z}\jacdet$.
We also report the number~(and percentage) of voxels that have at least one $\vert J^{}_{i}\vert \leq0$, denoted as `$\text{Any } \vert J^{}_{i}\vert \leq0$' in the tables.
We observe that in most of the cases, there are actually more than twice the number of voxels having $\vert J_i\vert \leq0$ for some finite difference than found using only the central difference.
The differences between `$\mathcal{D}^{0x}\mathcal{D}^{0y}\mathcal{D}^{0z}\jacdet \leq 0$' and `$\text{Any } \vert J^{}_{i}\vert  \leq 0$' highlight that errors in using the central difference approximation are very common in practice.

The proposed average non-diffeomorphic volume~(NDV) and its percentage are shown in the last three columns of Tables~\ref{t:ixi} and~\ref{t:l2r}.
For algorithms that impose strong regularization on the transformations~(\eg \textit{MIDIR}~\citep{qiu2021learning}, \textit{deedsBCV}~\citep{heinrich2015multi}, \textit{NiftyReg}~\citep{modat2010fast}, and \textit{SyN}~\citep{avants2008symmetric}), a small---even zero---NDV is observed.
However, for the deep learning methods that directly output deformation fields~\citep{kim2021cyclemorph, balakrishnan2019voxelmorph, chen2022transmorph, liu2022coordinate}, we usually have higher NDVs.
It is important to note that the results shown in Tables~\ref{t:ixi} and~\ref{t:l2r} do not reflect the accuracy of the algorithms, just the proportion of their deformation that is non-diffeomorphic.
\begin{table*}[!tb]
    \centering
    \caption{Results for the \texttt{Learn2Reg OASIS} dataset. For `$\mathcal{D}^{0x} \mathcal{D}^{0y} \mathcal{D}^{0z} \jacdet \leq 0$' and `Any $\vert J^{}_{i}\vert \leq0$' we report the mean number of voxels (voxel \#) over the 19 validation subject pairs and the corresponding standard deviation ($\pm$), as well as the percentage (\%) with respect to the brain mask.
    We also report our proposed measure of non-diffeomorphic space---\ie non-diffeomorphic volume~(NDV) in 3D---and the corresponding standard deviations and percentages. Methods are listed in the order in which they were published.
    \label{t:l2r}}
    \adjustbox{max width = 0.99 \textwidth}
    {
        \begin{tabular}{l p{0.25\textwidth} C{0.18\textwidth} C{0.18\textwidth} C{0.18\textwidth}
        }
            \toprule
            && $\mathcal{D}^{0x}\mathcal{D}^{0y}\mathcal{D}^{0z}\jacdet\leq0$
            &
            Any $\vert J^{}_{i}\vert \leq0$
            &
            \cellcolor{green!25}Proposed
            \\
            \cmidrule(lr){3-3} \cmidrule(lr){4-4} \cmidrule(lr){5-5}
            &&
            \# of voxels & \# of voxels & NDV
            \\
            && (\%) & (\%) & (\%)
            \\
            \midrule
            & & $205.5\pm331.1$ & $278.2\pm415.3$ & $85.0\pm204.1$
            \\
            \multirow{-2}{*}{\textit{SyN}} & \multirow{-2}{*}{\citep{avants2008symmetric}}
            & $(0.01\%)$ & $(0.02\%)$ & $(0.01\%)$
            \\
            \rowcolor{Gray} & & $40418.4\pm8991.3$ & $94343.8\pm18452.7$ & $18448.2\pm4319.0$
            \\
            \rowcolor{Gray}
            \multirow{-2}{*}{\textit{Voxelmorph}} & \multirow{-2}{*}{\citep{balakrishnan2019voxelmorph}}
            & $(2.84\%)$ & $(6.64\%)$ & $(1.30\%)$
            \\
            & & $31646.6\pm7609$ & $78533.7\pm16113.0$ & $13502.7\pm3779.9$
            \\
            \multirow{-2}{*}{\textit{Transmorph}} & \multirow{-2}{*}{\citep{chen2022transmorph}}
            & $(2.22\%)$ & $(5.52\%)$ & $(0.95\%)$
            \\
            \rowcolor{Gray} & & $22905.6\pm4142.3$ & $161071.8\pm18271.5$ & $8774.7\pm975.6$
            \\
            \rowcolor{Gray}
            \multirow{-2}{*}{\textit{im2grid}} & \multirow{-2}{*}{\citep{liu2022coordinate}}
            & $(1.61\%)$ & $(11.36\%)$ & $(0.62\%)$
            \\
        \bottomrule
        \end{tabular}
    }
\end{table*}

\section{Discussion}
\label{s:dis}
The Jacobian determinant of a transformation is a widely used measure in deformable image registration, but the details of its computation are often overlooked.
In this paper, we focused on the finite difference based approximation of $\jacdet$.
Contrary to what one might expect, the commonly used central difference based $\jacdet$ does not reflect if the transformation is diffeomorphic or not.
Our investigation shows that each of the finite difference approximated $\jacdet$ corresponds to the signed area of a triangle in 2D or the signed volume of a tetrahedron in 3D when the digital transformations are assumed to be piecewise linear.
Following this, we propose the definition of a digital diffeomorphism that allows diffeomorphisms---a concept in continuous domain---to be applied to digital transformations.
It solves several problems that are inherent in the central difference based $\jacdet$.
We further propose to use non-diffeomorphic volume to measure the irregularity of 3D transformations and non-diffeomorphic area for 2D transformations.
As demonstrated in Fig.~\ref{f:nda_demo} and Fig.~\ref{f:ixi}, our proposed approach measures the \emph{severity} of the irregularity whereas the commonly used central difference based $\jacdet$ is only a binary indicator of folding~(also with errors).
The transformation shown in Fig.~\ref{f:nda_demo}(b) is obviously more favorable than the one shown in Fig.~\ref{f:nda_demo}(c) in terms of regularity. As such, it is important for us to be able to draw distinctions between these two scenarios.

The non-diffeomorphic measures presented in Eqs.~\ref{e:ndv} and~\ref{e:nda} are averages of two choices of tetrahedralization of the volume.
It is possible, however, to tetrahedralize each cube between voxels based on the specific registration outcome, for example, to yield the minimum or maximum achievable non-diffeomorphic volume. These measures must be computed for the entire image domain since the brain mask is defined at voxel level.
For the \texttt{IXI} dataset, the average non-diffeomorphic volume for the best case is 33227~voxel$^3$ for \textit{Voxelmorph} and 31488~voxel$^3$ for \textit{Transmorph} while the average non-diffeomorphic Volume for the worse case is 41903~voxel$^3$ for \textit{Voxelmorph} and 40964~voxel$^3$ for \textit{Transmorph}.
In comparison, our proposed NDV is 37565~voxel$^3$ for \textit{Voxelmorph} and 36226~voxel$^3$ for \textit{Transmorph} for the entire image domain.
This result shows that, at least for these two algorithms, the non-diffeomorphic volume cannot be made substantially different from the average value that we specified in Eq.~\ref{e:ndv} by alternative selection of tetrahedralization.

The idea of discretizing the space using triangles or tetrahedrons has been presented before for regularizing deformable transformations and deformation-based volume change estimation~\citep{pai2016deformation, holland2011nonlinear, yushkevich2010bias}.
It was previously believed that computing the volume change in discretized space using tetrahedrons is more accurate than using Jacobian determinants because the latter involves finite difference approximation~\citep{yushkevich2010bias}.
Our analysis shows that discretizing the space using triangles or tetrahedrons is in fact the consequence of finite difference approximation of Jacobian determinants. 
Therefore, the two approaches are equivalent when the corresponding finite differences are chosen for a given partition of space.
Haber~\etal~\citep{haber2004numerical, haber2007image} and Burger~\etal~\citep{burger2013hyperelastic} used the volume of triangles in 2D or tetrahedrons in 3D as a regularization term for their registration algorithm.
Initially, Haber~\etal~proposed a hard equality constraint~\citep{haber2004numerical} that enforces the preservation of the discretized volume (or area in 2D) of every deformed box. 
Recognizing that this approach could not detect ``twists" (\emph{i.e.}, folding), they later proposed an inequality constraint~\citep{haber2007image} that calculates the volumes of the tetrahedrons to prevent such twists by imposing positive volumes of the tetrahedrons. 
However, their methods only account for a single combination of Jacobian determinants (as shown in Fig.~\ref{f:meshes}(a)), which is insufficient to guarantee a digital diffeomorphism.
Moreover, these previous works were motivated by the fact that calculating the volume of a ``twisted'' polygon or octahedron is difficult.
Our analysis on the central difference approximated $\jacdet$ explains why using a polygon or octahedron is inaccurate.

With respect to deep network based registration methods, many ensure that the output of their network is diffeomorphic, by adopting a scaling-and-squaring layer as the output layer~\citep{dalca2019unsupervised, chen2022transmorph, hoopes2021hypermorph}.
However, several works have reported that the scaling-and-squaring approach produces voxels with negative central difference based Jacobian determinant. 
The analysis presented in this paper explains why the scaling-and-squaring approach cannot guarantee a folding-free digital transformation.
Specifically, when composing two digital transformations, one of the transformation needs to be sampled at non-grid locations, which usually involves bilinear or trilinear interpolation.
These interpolation methods, however, are inconsistent with the piecewise linear transformation that is implicitly assumed by the finite difference based Jacobian determinant computation.
As a result, the sampling process can introduce folding and result in locations with negative Jacobian determinant.

For recent deep learning based registration methods, our digital diffeomorphism criteria has the potential to be used as a loss function to improve the smoothness and promote digitally diffeomorphic transformations~\citep{mok2020fast}. This is a promising direction for future research.

\backmatter


\section*{Statements and Declarations}
\noindent\textbf{Data availability}
The datasets analysed during the current study are available in the OASIS repository, \url{https://www.oasis-brains.org/} and the IXI repository, \url{https://brain-development.org/ixi-dataset/}

\noindent\textbf{Ethical approval declarations}
This work involved human subjects or animals in its research. 
Approval of all ethical and experimental procedures and protocols was granted by the relevant local Institutional Review Boards, and performed in line with the Declaration of Helsinki.

\noindent\textbf{Conflict of interest}
The authors have no competing interests to declare
that are relevant to the content of this article

\noindent\textbf{Funding}
This work was supported in part by the National Institute of Health~(NIH) National Eye Institute grant R01-EY032284~(PI: J.L. Prince).


\bibliography{refs}


\end{document}